\documentclass[12pt,reqno]{amsart}
\usepackage{amsmath}
\usepackage{latexsym}
\usepackage{amsfonts}
\usepackage{amssymb}
\usepackage{color}
\usepackage{bbm,dsfont}
\usepackage{graphicx}
\usepackage{textcomp}

\newtheorem{proposition}{Proposition}

\theoremstyle{definition}

\newtheorem{example}{Example}

\newcommand{\real}{\mathbb R} 
\newcommand{\complex}{\mathbb C} 
\newcommand{\half}{\tfrac{1}{2}} 
\newcommand{\mo}[1]{\left| #1 \right|} 

\newcommand{\hi}{\mathcal{H}} 
\newcommand{\eh}{\mathcal{E(H)}} 
\newcommand{\ip}[2]{\left\langle\,#1\,|\,#2\,\right\rangle} 
\newcommand{\kb}[2]{|#1\rangle\langle#2|} 
\newcommand{\no}[1]{\left\|#1\right\|} 
\newcommand{\id}{\mathbbm{1}} 

\newcommand{\ve}{\mathbf{e}} 
\newcommand{\vf}{\mathbf{f}} 
\newcommand{\vsigma}{\boldsymbol{\sigma}} 

\newcommand{\ehqubit}{\mathcal{E}(\complex^2)}

\newcommand{\jordan}{\circ}

\newcommand{\meet}{\wedge}

\begin{document}

\title[]{A simple sufficient condition for the coexistence of quantum effects}

\author[]{Teiko Heinosaari}

\address{\newline Teiko Heinosaari \newline
Turku Centre for Quantum Physics, Department of Physics and Astronomy
\newline
University of Turku, Finland}
\email{teiko.heinosaari@utu.fi}

\begin{abstract}
Two quantum effects are considered coexistent if they can be measured together. 
It is known that commutativity and comparability are sufficient but not necessary for the coexistence of two effects.
We unify those two conditions to a simple but more widely applicable sufficient condition.
\end{abstract}

\maketitle

\section{Introduction}\label{sec:intro}

Experimental events are in quantum theory described as effect operators (effects for short) acting on a complex Hilbert space $\hi$., i.e., selfadjoint operators $E:\hi\to\hi$ satisfying $0 \leq \ip{\psi}{E \psi}\leq 1$ for all unit vectors $\psi\in\hi$.
We denote by $\eh$ the set of all effects.
This is a convex set and endowed with the complementation $E^\bot=\id-E$.

One of the most important relations in $\eh$ is \emph{coexistence} \cite{FQMI83,SEO83}.
By definition, two effects $E$ and $F$ are coexistent if they can be measured together, i.e., are contained in the range of a single positive operator valued measure (POVM). 
The properties of the coexistence relation have been investigated from various aspects \cite{LaPu01,Lahti03,Gudder09,HeReSt08,Molnar01b}, and the coexistence relation has been characterized first in some special pairs of qubit effects \cite{Busch86,BuHe08,LiLiYuCh09} and later for all pairs of qubit effects \cite{StReHe08,YuLiLiOh10,BuSc10}.
However, in the general case the coexistence relation has remain rather abstract and unmanageable relation.
The purpose of this communication is present a simple sufficient condition for the coexistence of two effects.

Our investigation is organized as follows.
First, in Section \ref{sec:coex} we recall the two well-known sufficient conditions for coexistence: commutativity and comparability.
In Section \ref{sec:inf} we briefly review the peculiar nature of the order structure of $\eh$ and explain why the sufficient condition given by the infimum of two effects, although reasonable, is not very useful.
In Section \ref{sec:jordan} we present a sufficient condition for coexistence using Jordan products, and this generalizes commutativity.
In Section \ref{sec:ginf} we use the notion of generalized infimum to derive a sufficient condition for coexistence that covers both the commutativity and comparability criteria.
Finally, our investigation is summarized in Section \ref{sec:summary}.

\section{Coexistence, Commutativity and Comparability}\label{sec:coex}

In this section we recall some basic definitions and results, all that can be found in \cite{SEO83}.

By definition, effects $E_1,\ldots,E_n\in\eh$ are coexistent if there exists a POVM $A$ such that $A(X_1)=E_1,\ldots,A(X_n)=E_n$ for some outcome sets $X_1,\ldots,X_n$.
If we are considering a pair of effects, then the definition of coexistence can be rewritten in two simple ways.
First, it is easy to see that two effects $E$ and $F$ occur in the range of some POVM if and only if there exist four effects $G_{11}$, $G_{12}$, $G_{21}$ and $G_{22}$ such that
\begin{align}\label{eq:coex:4}
\begin{split}
& G_{11} + G_{12} = E \, , \quad G_{11} + G_{21} = F \\
& G_{11} + G_{12} + G_{21} + G_{22} = \id \, .
\end{split}
\end{align}
Let us notice that other operators in \eqref{eq:coex:4} are determined from $E,F$ and $G_{11}$, e.g., $G_{12}=E-G_{11}$.

There is a natural partial order on $\eh$; $E\leq F$ if $\ip{\psi}{E \psi} \leq \ip{\psi}{F \psi}$ for all $\psi\in\hi$.
Clearly, $E\leq F$ if and only if there exists an effect $E'$ such that $E+E'=F$.
We thus obtain another equivalent formulation of coexistence: $E$ and $F$ are coexistent if and only if there exists an effect $G$ such that
\begin{align}\label{eq:coex:G}
G \leq E \, , \quad G \leq F \, , \quad G + \id \geq E + F \, .
\end{align}

It is well-known that two effects are coexistent if they commute.
Namely, if $E,F$ are two effects and 
\begin{equation}\label{cond:commu}
\tag{COMMU}
EF=FE
\end{equation}
then we can choose 
\begin{equation}\label{eq:coex:com}
G_{11}=EF \, , \quad G_{12}=EF^\bot \, , \quad G_{21}=E^\bot F \, , \quad G_{22}=E^\bot F^\bot
\end{equation}
 and and the equations in \eqref{eq:coex:4} are satisfied.
The commutativity of $E$ and $F$ guarantees that all the four operators in \eqref{eq:coex:com} are effects since, for instance, $EF=\sqrt{E}F\sqrt{E}\geq 0$.

Apart from commutativity, there is also another well-known sufficient condition for the coexistence of $E$ and $F$, their comparability.
If $E\leq F$ or $F \leq E$, then we can choose either $G=E$ or $G=F$ and \eqref{eq:coex:G} clearly holds.
To fully benefit from this fact, we notice that if \eqref{eq:coex:4} holds, then $G_{21}+G_{22}=E^\bot$ and $G_{12}+G_{22}=F^\bot$.
Hence, $E$ and $F$ are coexistent if and only if $E,E^\bot,F,F^\bot$ are coexistent.
This also means that the following are equivalent:
\begin{itemize}
\item[(i)] $E$ and $F$ are coexistent
\item[(ii)] $E^\bot$ and $F$ are coexistent
\item[(iii)] $E$ and $F^\bot$ are coexistent
\item[(iv)] $E^\bot$ and $F^\bot$ are coexistent
\end{itemize}
Thus, taking the complement effects into account, we recover the following well-known sufficient condition for coexistence:
two effects $E$ and $F$ are coexistent if
\begin{equation}\label{cond:comp}
\tag{COMP}
E \leq F \quad \textrm{or} \quad F\leq E \quad \textrm{or} \quad E \leq F^\bot \quad \textrm{or} \quad F^\bot \leq E
\end{equation}
Two effects that satisfy \eqref{cond:comp} are often called \emph{trivially coexistent}.

Obviously, the coexistence of some pairs of effects, such as $\half \id$ and $\frac{1}{2} E$, follow from both \eqref{cond:commu} and \eqref{cond:comp}, but it is easy to see that these conditions have also separate areas of applicability.
For instance, let $E,F\in\eh$ be two non-commuting effects and fix a number $0<t\leq\half$.
Then also the effects 
\begin{equation}\label{eq:half}
t E + (1-t) \id \quad \textrm{and} \quad t F + (1-t) \id
\end{equation}
 are non-commuting, but this latter pair satisfies \eqref{cond:comp} since
\begin{equation}
(t F + (1-t) \id)^\bot \leq t E + (1-t) \id \quad \Leftrightarrow \quad (2t-1) \id \leq t(E+F)  \, .
\end{equation}
To see an example where \eqref{cond:commu} holds but \eqref{cond:comp} not, fix an effect $E$ such that $E \nleq \half \id$ and $E \ngeq \half \id$ (e.g. a non-trivial projection), and two numbers $0<s,t<1$.
The effects 
\begin{equation}
sE + (1-s) E^\bot \quad \textrm{and} \quad t E + (1-t) E^\bot
\end{equation}
 always commute, but \eqref{cond:comp} holds only if $s=t$ or $s=1-t$.

\section{Infimum}\label{sec:inf}

There is a natural way to generalize \eqref{cond:comp} and obtain a new sufficient condition for coexistence using the infimum of two effects.
In a closer look this approach turns out to be very restricted, but it well demonstrates the delicate nature of the coexistence relation and we therefore look briefly look at it.

Suppose that the infimum of two effects $E$ and $F$ exists and is denoted by $E \meet F$.
By definition, $E \meet F$ is the greatest of all effects $C$ satisfying $C\leq E$ and $C \leq F$.
Therefore, whenever there is an effect $G$ satisfying \eqref{eq:coex:G}, then also $E \meet F$ satisfies \eqref{eq:coex:G}.
We thus conclude that if $E \meet F$ exists, then $E$ and $F$ are coexistent if and only if 
\begin{equation}\label{eq:coex:inf}
E \meet F \geq E + F - \id \, .
\end{equation}
We should note that even if $E \meet F$ would exist, \eqref{eq:coex:inf} is a useful sufficient criterion for coexistence only if $E \meet F$ can be expressed in some explicit form.
We recall some results related to the existence and form of the infimum \cite{MoGu99,GhGuJo05}.
For $E,F\in\eh$, we denote by $P_{E,F}$ the projection onto the closer of $ran(\sqrt{E}) \cap ran(\sqrt{F})$. 
The infimum of a projection and an effect always exists  and has an explicit expression.
Hence, we can calculate $E \meet P_{E,F}$ and $F \meet P_{E,F}$.
Then, $E \meet F$ exists if and only if $(E \meet P_{E,F}) \meet (F \meet P_{E,F})$ exists, and in this case 
\begin{equation}
E \meet F =(E \meet P_{E,F}) \meet (F \meet P_{E,F}) \, .
\end{equation}
These results do not yet give an explicit form for $E \meet F$, but for $\dim \hi <\infty$ a complete solution is known \cite{MoGu99}.
Then $E \meet F$ exists if and only if $E \meet P_{E,F}$ and $F \meet P_{E,F}$ are comparable.
If this is the case, $E \meet F$ is the smaller of $E \meet P_{E,F}$ and $F \meet P_{E,F}$.
We conclude that for $\dim \hi <\infty$ the infimum can be calculated whenever it exists, but for $\dim \hi = \infty$ the explicit form for the infimum seems to be lacking.

Taking into account the complement effects, we can formulate the following sufficient condition for coexistence. 

\begin{proposition}\label{prop:inf}
Two effects $E$ and $F$ are coexistent if one of the following conditions hold:
\begin{equation}\label{cond:inf}
\tag{INF}
\begin{split}
& E \meet F \quad  \textrm{exists and} \quad E \meet F \geq E + F - \id \\
& E \meet F^\bot \quad  \textrm{exists and} \quad E \meet F^\bot \geq E + F^\bot - \id \\
& E^\bot \meet F \quad  \textrm{exists and} \quad E^\bot \meet F \geq E^\bot + F - \id \\
& E^\bot \meet F^\bot \quad  \textrm{exists and} \quad E^\bot \meet F^\bot \geq E^\bot + F^\bot - \id
\end{split}
\end{equation}
\end{proposition}

It is easy to see that \eqref{cond:inf} is a generalization of \eqref{cond:comp}.
If, for instance, $E\leq F$, then $E \meet F = E$ and $E \meet F \geq E + F - \id$ is equivalent to $\id \geq F$, and is therefore true.
But \eqref{cond:inf} is not a generalization of \eqref{cond:commu} since the infimum of two commuting effects may not exist \cite{LaMa95}.

In the following example we demonstrate the use of \eqref{cond:inf}.

\begin{example}\label{ex:inf}
Suppose that $E,F$ are effects and $E$ is a multiple of a one-dimensional projection. 
We will show that $E$ and $F$ are coexistent if and only if $ran(E)\subseteq ran(\sqrt{F})$ or $E+F\leq \id$.

We can write $E=e\kb{\psi}{\psi}$ for some unit vector $\psi\in\hi$ and number $0 < e \leq 1$.
Clearly, $ran(\sqrt{E})=ran(E)=\complex \psi$.
There are two alternatives: either $\psi\notin ran(\sqrt{F})$ or $\psi \in ran(\sqrt{F})$.
In terms of the projection $P_{E,F}$ this means that either $P_{E,F}=0$ or $P_{E,F}=\kb{\psi}{\psi}$, respectively.

Let us first assume that $P_{E,F}=0$.
Then $E \meet P_{E,F}=F \meet P_{E,F}=0$ and therefore $E \meet F$ exists and $E \meet F =0$.
The coexistence condition $E \meet F \geq E+F-\id$ is equivalent to $E+F \leq \id$.

Let us then assume that $P_{E,F}=\kb{\psi}{\psi}$.
We have
\begin{align}
E \meet P_{E,F} = \ip{\psi}{E\psi} \kb{\psi}{\psi} \, , \qquad F \meet P_{E,F} = \ip{\psi}{F\psi} \kb{\psi}{\psi} \, .
\end{align}
It follows that $E \meet F$ exists and 
\begin{align}
E \meet F = \min \{Ê\ip{\psi}{E\psi}, \ip{\psi}{F\psi} \} \kb{\psi}{\psi} \, .
\end{align}
It is straightforward to verify that the coexistence condition $E \meet F \geq E + F - \id$ is always satisfied.
\end{example}

The limited applicability of \eqref{cond:inf} originates from the fact that the infimum of two effects exists only rarely.
This has become clear from various examples presented in earlier studies on the infimum.
For instance, for $\dim \hi=2$ the infimum of two effects $E$ and $F$ exists if and only if $E$ and $F$ are comparable, or one of them is a multiple of a $1$-dimensional projection \cite{GuGr96}.
As another example, if effects $E$ and $F$ are invertible operators, then $E \meet F$ exists if and only if $E$ and $F$ are comparable \cite{YuDu06}.

\section{Jordan product}\label{sec:jordan}

We will next seek a generalization of commutativity into a more widely applicable sufficient condition for coexistence.
For $E,F\in\eh$, we denote 
\begin{equation*}
E \jordan F= \half ( EF + FE) \, .
\end{equation*}
This is called the \emph{Jordan product} of $E$ and $F$.
It is clear that $E\jordan F$ is always a selfadjoint operator and
\begin{equation*}
E\jordan F \leq \no{E}\no{F} I \leq I \, .
\end{equation*}
However, $E\jordan F$ need not be a positive operator.

It is easy to verify that 
\begin{align*}
& E \jordan F + E \jordan F^\bot = E \, , \quad E \jordan F + E^\bot \jordan F = F \\
& E \jordan F + E^\bot \jordan F + E \jordan F^\bot + E^\bot \jordan F^\bot =I \, .
\end{align*}
Hence, by choosing $G_{11}=E \jordan F$, $G_{12}=E \jordan F^\bot$, $G_{21}=E^\bot \jordan F$ and $G_{22}=E^\bot \jordan F^\bot$, then \eqref{eq:coex:4} holds for any choice of $E,F\in\eh$.
The remaining step for the coexistence of $E$ and $F$ is to guarantee that these four operators are positive.  
We thus conclude the following sufficient criterion for coexistence.

\begin{proposition}\label{prop:jordan}
Two effects $E$ and $F$ are coexistent if
\begin{equation}\label{cond:jor}
\tag{JOR}
E \jordan F\geq 0 \ \mathrm{and} \ E^\bot \jordan F\geq 0 \ \mathrm{and} \ E \jordan F^\bot\geq 0 \  \mathrm{and} \  E^\bot \jordan F^\bot\geq 0 \, .
\end{equation}
\end{proposition}

This condition is a generalization of the commutativity condition \eqref{cond:commu}.
Namely, if $EF=FE$, then $E \jordan F = EF=\sqrt{E}F\sqrt{E}\geq0$.
Therefore, \eqref{cond:jor} holds whenever $EF=FE$.
But \eqref{cond:jor} need not hold if \eqref{cond:comp} holds.
For instance, let $\psi,\varphi\in\hi$ be two unit vectors such that $r:=\ip{\psi}{\varphi}\in\real$, $0<r<1$, and choose $E=\half \kb{\psi}{\psi}, F=\half \kb{\varphi}{\varphi}$.
Then $E+F \leq \id$, but $E \circ F \ngeq 0$ since $E \circ F = \frac{r}{8} ( \kb{\psi}{\varphi} + \kb{\varphi}{\psi})$
and the eigenvalues of this operator are $\frac{r}{8}(1+r)>0$ and $\frac{r}{8}(r-1)<0$.

In the following example we will see that  \eqref{cond:jor} covers more cases than just commutative pairs.

\begin{example}
Let $d < \infty$, $\hi=\complex^d$ and fix an orthonormal basis $\{\varphi_x\}_{x=0}^{d-1}$ for $\hi$.
We define a unit vector $\psi_0\in\hi$ as
\begin{equation*}
\psi_0 = \frac{1}{\sqrt{d}} \sum_{x=0}^{d-1} \varphi_x \, .
\end{equation*}
Then, we fix a number $0\leq\lambda\leq 1$ and define a pair of effects $E$ and $F$ by
\begin{align}\label{eq:mub-effects}
E=\lambda \kb{\varphi_0}{\varphi_0} + (1-\lambda) \frac{1}{d} \id \, , \quad F=\lambda \kb{\psi_0}{\psi_0} + (1-\lambda) \frac{1}{d} \id \, .
\end{align}
These two effects commute only if $\lambda=0$.
As shown in \cite{CaHeTo12}, $E$ and $F$ are guaranteed to be coexistent if
\begin{align}
\lambda \leq \frac{1}{2}+ \frac{\sqrt{d}-1}{2(d-1)} \equiv \lambda_{\mathrm{MAX}}(d)  \, .
\end{align}
We have
\begin{align}
& E \jordan F =\frac{\lambda^2}{2\sqrt{d}} \bigl( \kb{\varphi_0}{\psi_0} +  \kb{\psi_0}{\varphi_0} \bigr)+ \\
& \frac{\lambda(1-\lambda)}{d} \bigl( \kb{\varphi_0}{\varphi_0} + \kb{\psi_0}{\psi_0} \bigr)+ \frac{(1-\lambda)^2}{d^2} \id \, ,
\end{align}
and the requirement $E \jordan F\geq 0$ leads to the inequality
\begin{align}\label{eq:mub-jordan}
\lambda \leq \frac{2(1+\sqrt{2})d-4}{d^2+4d-4} \equiv \lambda_{\mathrm{JOR}}(d)  \, .
\end{align}
A similar calculation on the three other operator inequalities in \eqref{cond:jor} shows that they are less restrictive than $E \jordan F\geq 0$, hence \eqref{cond:jor} is equivalent to \eqref{eq:mub-jordan}.
It is straightforward to confirm that $\lambda_{\mathrm{JOR}}(2)=\lambda_{\mathrm{MAX}}(2)$ and $0<\lambda_{\mathrm{JOR}}(d) < \lambda_{\mathrm{MAX}}(d)$ whenever $d\geq 3$.
We conclude that \eqref{cond:jor} is a more general sufficient condition for coexistence than commutativity.
\end{example}

Finally, we note that Proposition \ref{prop:jordan} has a generalization that gives a sufficient condition for the coexistence of any finite number of effects. 
For an effect $E$, we denote $E^{(1)}=E$ and $E^{(2)}=E^\bot$.
Let $E_1,\ldots, E_n$ be a finite collection of effects.
We define
\begin{align}\label{eq:jordan:gen}
G_{i_1 i_2\cdots i_n} & =\frac{1}{n!} \bigl( E^{(i_1)}_1 E^{(i_2)}_2 \cdots E^{(i_n)}_n+E^{(i_2)}_2 E^{(i_1)}_1 \cdots E^{(i_n)}_n  \\
& \quad + \textrm{all other permutations} \bigr) \, . \nonumber 
\end{align}
It is easy to verify that
\begin{equation}
\sum_{i_2, \cdots, i_n \in \{1,2\}}  G_{1 i_2\cdots i_n} = E_1
\end{equation}
and similarly for other indices.
We thus conclude the following generalization of Proposition \ref{prop:jordan}.

\begin{proposition}\label{prop:jordan:gen}
Effects $E_1,\ldots, E_n$ are coexistent if $G_{i_1 i_2\cdots i_n}\geq 0$ for all $i_1,i_2, \cdots, i_n \in \{1,2\}$.
\end{proposition}

\begin{example}
Let
\begin{align*}
E_1= \half ( \id + \ve_1 \cdot \vsigma ) \, , \quad  E_2= \half ( \id + \ve_2 \cdot \vsigma ) \, , \quad  E_3= \half ( \id + \ve_3 \cdot \vsigma ) \, .
\end{align*}
for some orthogonal vectors $\ve_1,\ve_2,\ve_3$ with $\no{\ve_i}\leq 1$.
The smallest eigenvalue of each operator $G_{111}, \ldots, G_{222}$ is $\frac{1}{8} ( 1- \sqrt{\no{\ve_1}^2+\no{\ve_2}^2+\no{\ve_3}^2} )$.
Therefore, we recover the sufficient condition found in \cite{Busch86}: the effects $E_1,E_2,E_3$ are coexistent if  
\begin{equation}
\no{\ve_1}^2+\no{\ve_2}^2+\no{\ve_3}^2 \leq 1 \, .
\end{equation}
A similar calculation can be done for a non-orthogonal triplet, but then the positivity conditions for $G_{i_1i_2i_3}$ are more complicated.
\end{example}

\section{Generalized infimum}\label{sec:ginf}

Since \eqref{cond:inf} has a very limited area of applicability, we will seek another way to generalize \eqref{cond:comp} to a wider sufficient condition for coexistence.
For $E,F\in\eh$, we denote
\begin{equation}
E \sqcap F = \half ( E+F - \mo{E-F} ) 
\end{equation}
and call this operator \emph{generalized infimum} of $E$ and $F$.
The operator $E \sqcap F $ has many useful properties similar to the infimum. 
For instance, if $E\leq F$, then $E \sqcap F=E$.
The properties of $E \sqcap F $ have been investigated in various works \cite{YuDu06,Topping65,Gudder96,LiSuCh07}, but here we only need some basic facts.

The operator $E \sqcap F$ is selfadjoint and satisfies
\begin{equation}\label{eq:ginf:below}
E \sqcap F\leq E \qquad \textrm{and} \qquad E \sqcap F \leq F \, .
\end{equation}
The validity of these operator inequalities can be seen as follows \cite{Gudder96}.
First, we have $E-F \leq \mo{E-F}$.
It follows that $E+F-\mo{E-F} \leq 2F$, which means that $E \sqcap F\leq F$.
Since $\mo{E-F}=\mo{F-E}$, a similar argument gives $E \sqcap F\leq E$. 

We look for a sufficient condition for coexistence. 
If we want that $G=E \sqcap F$ is an effect and satisfies \eqref{eq:coex:G}, then  \eqref{eq:ginf:below} holds for all pairs but we need to additionally require that
\begin{equation}\label{eq:ginf:pos}
 E\sqcap F \geq 0 \qquad \textrm{and} \qquad E \sqcap F \geq E + F - \id \, .
\end{equation}
The latter operator inequality in \eqref{eq:ginf:pos} can be written in an equivalent form:
\begin{align*}
& E \sqcap F \geq E + F - \id \\
\Leftrightarrow \quad & \half ( E + F - \mo{E - F} ) \geq E + F - \id \\
\Leftrightarrow \quad &   \half( 2\id - E - F +\mo{E-F} ) \geq 0 \\
\Leftrightarrow \quad &   E^\bot \sqcap F^\bot \geq 0 \, . 
\end{align*}
Finally, taking into account the complement effects we conclude the following sufficient condition for coexistence.

\begin{proposition}\label{prop:ginf}
Two effects $E$ and $F$ are coexistent if
\begin{equation}\label{cond:ginf}
\tag{GINF}
\begin{split}
\bigl( E\sqcap F \geq 0  \  \mathrm{and} \    E^\bot \sqcap F^\bot \geq 0 \bigr) \quad \mathrm{or} \quad \bigl(
E^\bot\sqcap F \geq 0  \  \mathrm{and} \    E \sqcap F^\bot \geq 0 \bigr) \, .
\end{split}
\end{equation}
\end{proposition}

It is clear that \eqref{cond:comp} implies \eqref{cond:ginf}.
For instance, if $E\leq F$, then $E \sqcap F = E\geq 0$ and $E^\bot \sqcap F^\bot = F^\bot \geq 0$. 
More interestingly, \eqref{cond:jor} implies \eqref{cond:ginf}, as we next prove.

\begin{proposition}
Let $E,F\in\eh$.
If \eqref{cond:jor} holds, then \eqref{cond:ginf} also holds.
\end{proposition}

\begin{proof}
We first note that $E \jordan F \geq 0$ if and only if $(E-F)^2 \leq (E+F)^2$.
Since the square root is an operator monotone function (see e.g. \cite{CALGEBRAS79}), the latter operator inequality implies that $\mo{E-F} \leq E+F$.
 This is equivalent to $E \sqcap F \geq 0$.
\end{proof}

Since \eqref{cond:commu} implies \eqref{cond:jor}, we conclude that \eqref{cond:ginf} covers both \eqref{cond:comp} and \eqref{cond:commu}.
But \eqref{cond:ginf} has a wider area of applicability than \eqref{cond:comp} and \eqref{cond:commu} together.
This is demonstrated in the following example.

\begin{example}\label{ex:busch}
Let $\hi=\complex^2$ and denote $\vsigma=(\sigma_x,\sigma_y,\sigma_z)$, where $\sigma_x,\sigma_y$ and $\sigma_z$ are the usual Pauli matrices.
We consider two effects $E$ and $F$ of the form
\begin{equation}
E=\half (\id + \ve \cdot \vsigma) \, , \qquad F=\half (\id + \vf \cdot \vsigma)
\end{equation}
for some $\ve,\vf\in\real^3$ with $\no{\ve}\leq 1,\no{\vf}\leq 1$.
It was first proved in \cite{Busch86} that $E$ and $F$ are coexistent if and only if 
\begin{equation}\label{eq:busch}
\no{\ve+\vf} + \no{\ve-\vf} \leq 2  \, .
\end{equation}
A direct calculation shows that the smallest eigenvalue of $E \sqcap F$, $E\sqcap F^\bot$, $E^\bot \sqcap F$ and $E^\bot \sqcap F^\bot$ is $\frac{1}{4} ( 2- \no{\ve-\vf}-\no{\ve+\vf})$.
We conclude that \eqref{cond:ginf} is equivalent to \eqref{eq:busch}, hence necessary and sufficient for the coexistence of $E$ and $F$.
\end{example}

The following example shows that \eqref{cond:ginf} is not a necessary condition for coexistence.

\begin{example}
Let 
\begin{equation}
E=\half (\id + \ve \cdot \vsigma) \, , \quad F=\half ( \beta \id + \vf \cdot \vsigma)
\end{equation}
for some for some \emph{orthogonal} vectors $\ve,\vf\in\real^3$ with $\no{\ve}\leq 1,\no{\vf}\leq 1$, and $\no{\vf} \leq \beta \leq 2-\no{\vf}$.
It was first proved in \cite{LiLiYuCh09} that $E$ and $F$ are coexistent if and only if 
\begin{equation}\label{eq:liu}
2 \no{\ve} \leq \sqrt{\beta^2 -\no{\vf}^2} + \sqrt{(2-\beta)^2-\no{\vf}^2}  \, .
\end{equation}
The choices $\no{\ve}=\no{\vf}=2/3$ and $\beta=3/4$ satisfy \eqref{eq:liu}.
However, with these values of parameters the operators $E \sqcap F$ and $E^\bot \sqcap F$ have a negative eigenvalue. 
We conclude that \eqref{cond:ginf} is not a necessary condition for coexistence.
\end{example}

Finally, let us consider the relationship between \eqref{cond:ginf} and \eqref{cond:inf}.
It is easy to see that \eqref{cond:ginf} does not imply \eqref{cond:inf}.
Namely, \eqref{cond:ginf} can hold even if none of $E \meet F$, $E \meet F^\bot$, $E^\bot \meet F$, $E^\bot \meet F^\bot$ exist.
This can be seen from Example \ref{ex:busch} and the result mentioned earlier \cite{GuGr96}: the infimum of $E,F\in \ehqubit$ exists if and only if $E$ and $F$ are comparable or one of them is a multiple of a one-dimensional projection. 

An interesting question is whether \eqref{cond:inf} implies \eqref{cond:ginf}, and if not, to find a sufficient condition that covers both of them.
We cannot offer a solution at the moment, but there is some indication that \eqref{cond:inf} may imply (GINF).
It was proved in \cite{Gudder96} that if $E\sqcap F \geq 0$ and $E \meet F$ exists, then $E \meet F = E\sqcap F$.
In this case, the operator inequalities $E \meet F \geq E + F - \id$ and $E^\bot \sqcap F^\bot \geq 0$ are equivalent.
We conjecture that \eqref{cond:inf} implies $E\sqcap F \geq 0$, and this would then mean that \eqref{cond:inf} implies (GINF).
Related to this conjecture, we recall that for an effect $E$ and a projection $P$, $EP=PE$ if and only if $E \sqcap P \geq 0$ \cite{LiSuCh07}.
On the other hand, since $E \meet P$ always exists \cite{MoGu99}, we see that \eqref{cond:inf} holds if and only if $EP=PE$.
Therefore, \eqref{cond:inf} and \eqref{cond:ginf} are equivalent conditions for two effects if one of them is a projection.

\section{Summary}\label{sec:summary}

There are two well-known sufficient criteria for coexistence, comparability and commutativity.
We have introduced a new sufficient condition that is a generalization of both comparability and commutativity.
For two effects $E$ and $F$, the generalized infimum is defined as $E \sqcap F = \half ( E + F - \mo{E-F})$.
The new sufficient condition is
\begin{equation*}
\begin{split}
\bigl( E \sqcap F \geq 0  \  \mathrm{and} \    E^\bot \sqcap F^\bot \geq 0 \bigr) \quad \mathrm{or} \quad \bigl( E^\bot \sqcap F \geq 0  \  \mathrm{and}  \    E \sqcap F^\bot \geq 0 \bigr) \, .
\end{split}
\end{equation*}
We have demonstrated that this new condition does not only combine comparability and commutativity, but covers also other pairs of coexistent effects.
It is, however, not a necessary condition for coexistence.

There are some natural questions related to the above sufficient condition.
Is it possible to write it in an equivalent but simpler way, for instance by dropping some of the operator inequalities?
Is the infimum condition \eqref{cond:inf} introduced in Sec. \ref{sec:inf} covered by the new condition?

\section*{Acknowledgements}

This work has been supported by the Academy of Finland (grant no. 138135).

\end{document}